\documentclass[runningheads]{llncs}
\usepackage[margin=1in, bottom=1in, top=1in]{geometry} 
\usepackage{amsmath, amssymb, amstext}
\usepackage{fancyhdr}
\usepackage{mathtools}
\usepackage{algorithm}
\usepackage{algpseudocode}
\usepackage{theorem}
\usepackage{tikz}
\usepackage{tkz-berge}
\usepackage{authblk}
\usepackage[colorlinks]{hyperref}
\usepackage{enumerate}  
\usepackage{bbm}
\usepackage{mathbbol}
\usepackage{comment}
\usepackage{wrapfig}

\newcommand{\N}{\mathbb{N}} \newcommand{\R}{\mathbb{R}}

\newcommand{\0}{\mathbb{0}} 
\newcommand{\1}{\mathbb{1}}

 \newcommand{\cF}{\mathcal{F}}

 \newcommand{\cL}{\mathcal{L}}
\newcommand{\cM}{\mathcal{M}} 
\newcommand{\cO}{\mathcal{O}}

\newcommand{\todo}[1]{{\tiny\color{red}$\scriptscriptstyle\circ$}{\marginpar{\flushleft\baselineskip=2pt\tiny\sf$\bullet$\color{red}#1}}}

\newcommand{\uni}{uni}
\newcommand{\Muni}{\cM_{\uni}}

\DeclareMathOperator{\suppOp}{supp}
\newcommand{\supp}{\suppOp}

\DeclareMathOperator{\convOp}{conv}
\newcommand{\conv}{\convOp}

\newcommand{\fix}{\mathrm{Fix}}

\DeclareMathOperator{\exOp}{excess}
\newcommand{\ex}{\exOp}

\DeclareMathOperator{\symOp}{diff}
\newcommand{\sym}{\symOp}

\DeclareMathOperator{\parent}{parent}
\DeclareMathOperator{\optop}{top}
\DeclareMathOperator{\excess}{excess}
\DeclareMathOperator{\opspan}{span}

\newcommand{\floor}[1]{\ensuremath{\left\lfloor#1\right\rfloor}}

\usepackage{graphicx}
\begin{document}
\title{Computing the Nucleolus of Weighted Cooperative Matching Games in Polynomial Time \thanks{This work was done in part while the second author was visiting the Simons Institute for the Theory of Computing. Supported by DIMACS/Simons Collaboration on Bridging Continuous and Discrete Optimization through NSF grant \#CCF-1740425.}\thanks{We acknowledge the support of the Natural Sciences and Engineering Research Council of Canada (NSERC). Cette recherche a \'et\'e financ\'ee par le Conseil de recherches en sciences naturelles et en g\'enie du Canada (CRSNG).}}
%
%
\author{Jochen K\"{o}nemann\inst{1} \and
Kanstantsin Pashkovich\inst{2} \and
Justin Toth\inst{1}}
\authorrunning{J. K\"{o}nemann et al.}
%
\institute{University of Waterloo, Waterloo ON N2L 3G1,  Canada 
\email{\{jochen, wjtoth\}@uwaterloo.ca}\and
University of Ottawa, Ottawa ON K1N 6N5, Canada \email{kpashkov@uottawa.ca}}
\maketitle              
\begin{abstract}
  We provide an efficient algorithm for computing the nucleolus for an
  instance of a weighted cooperative matching game. This resolves a
  long-standing open question posed in [Faigle, Kern, Fekete,
  Hochst\"{a}ttler, Mathematical Programming, 1998].
\end{abstract}

\section{Introduction}\label{sec:intro}

Imagine a network of players that form partnerships to generate value. For example, maybe a tennis league pairing players to play exhibition matches~\cite{Biro2012}, or people making trades in an exchange network~\cite{SS71}. These are examples of what are called \emph{matching games}. In a (weighted) matching game, we are given a graph $G=(V,E)$, weights $w: E \rightarrow \R_{\geq 0}$, the player set is the set $V$ of nodes of $G$, and $w(uv)$ denotes the value earned when $u$ and $v$ collaborate. Each coalition $S\subseteq V$ is assigned a \emph{value} $\nu(S)$ so that $\nu(S)$ is equal to the value of a maximum weight matching on the induced subgraph $G[S]$. The special case of matching games where $w=\1$ is the all-ones vector, and $G$ is bipartite is called an {\em assignment game} and was introduced in a classical paper by Shapley and Shubik~\cite{SS71}, and was later generalized to general graphs by Deng, Ibaraki, and Nagamochi~\cite{DIN99}.  

We are interested in what a fair redistribution of the total value $\nu(V)$ to the players in the network looks like. The field of \emph{cooperative game theory} gives us the language to make this question formal. A vector $x \in \R^V$ is called an {\em allocation} if $x(V)=\nu(V)$ (where we use $x(V)$ as a short-hand for $\sum_{i \in V}x(i)$ as usual). Given such an allocation, we let $x(S)-\nu(S)$ be the {\em excess} of coalition $S \subseteq V$. This quantity can be thought of as a measure of the satisfaction of coalition $S$. A fair allocation should maximize the bottleneck excess, i.e. maximize the minimum excess, and this can be  accomplished by an LP:
\begin{align*} \label{eq:coalition_leastcore}
  \max ~~ & \varepsilon \tag{$P$} \\
  \text{s.t.} ~~ & x(S) \geq \nu(S) + \varepsilon \qquad \text{for
                          all}\quad S \subseteq V \\
&  x(V) = \nu(V)\\
& x \geq \0.
\end{align*}
Let $\varepsilon^*$ be the optimum value of
\eqref{eq:coalition_leastcore}, and define $P(\varepsilon^*)$ to be
the set of allocations $x$ such that $(x,\varepsilon^*)$ is feasible
for \eqref{eq:coalition_leastcore}. The set $P(\varepsilon^*)$ is
known as the {\em leastcore}~\cite{MPS79} of the given cooperative
game, and the special case when $\varepsilon^* = 0$, $P(0)$ is the
well-known {\em core}~\cite{Gi59} of $(V,\nu)$. Intuitively,
allocations in the core describe payoffs in which no coalition of
players could profitably deviate from the \emph{grand coalition} $V$.

Why stop at maximizing the bottleneck excess? Consider an allocation
which, subject to maximizing the smallest excess, maximizes the second
smallest excess, and subject to that maximizes the third smallest
excess, and so on. This process of successively optimizing the excess
of the worst-off coalitions yields our primary object of interest, the
\emph{nucleolus}. For an allocation $x\in \R^V$, let
$\theta(x) \in \R^{2^V-2}$ be the vector obtained by sorting the list
of excess values $x(S) - \nu(S)$ for any
$\varnothing \neq S \subset V$ in non-decreasing order \footnote{It is
  common within the literature, for instance in~\cite{Kern2003}, to
  exclude the coalitions for $S = \varnothing$ and $S = V$ in the
  definition of the nucleolus. On the other hand, one could also
  consider the definition of the nucleolus with all possible
  coalitions, including $S = \varnothing$ and $S = V$. We note that
  the two definitions of the nucleolus are equivalent in all instances
  of matching games except for the trivial instance of a graph
  consisting of two nodes joined by a single edge.}. The {\em
  nucleolus}, denoted $\eta(V,\nu)$ and defined by
Schmeidler~\cite{schmeidler1969nucleolus}, is the unique allocation
that lexicographically maximizes $\theta(x)$:
 $$\eta(V,\nu) := \text{arg lex max}\{ \theta(x): x \in P(\varepsilon^*)\}.$$
We refer the reader to Appendix \ref{appendix:example} for an example
instance of the weighted matching game with its nucleolus. 
We now have sufficient terminology to state our main result:

\begin{theorem}\label{th:main-result}
Given a graph $G=(V,E)$ and weights $w:E\rightarrow \R$, the nucleolus $\eta(V,\nu)$ of the corresponding weighted matching game can be computed in polynomial time. 
\end{theorem}

Despite its intricate definition the concept of the nucleolus is surprisingly ancient. Its history can be traced back to a discussion on bankruptcy division in the Babylonian Talmud~\cite{AM85}. Modern research interest in the nucleolus stems not only from its geometric beauty~\cite{MPS79}, or several practical applications (e.g., see~\cite{BST05,Le84}), but from the strange way problems of computing the nucleolus fall in the complexity landscape, seeming to straddle the NP vs P boundary.

Beyond being one of the most fundamental problems in combinatorial optimization, starting with the founding work of Kuhn on the Hungarian method for the assignment problem~\cite{kuhn1955hungarian}, matching problems have historically teetered on the cusp of hardness. For example, prior to Edmonds' celebrated Blossom Algorithm~\cite{edmonds1965paths,edmonds1965maximum} it was not clear whether Maximum Matching belonged in P. For another example, until Rothvo{\ss}' landmark result~\cite{rothvoss2017matching} it was thought that the matching polytope could potentially have sub-exponential extension complexity. In cooperative game theory, matchings live up to their historical pedigree of representing a challenging problem class. The long standing open problem in this area was whether the nucleolus of a weighted matching game instance can be computed in polynomial time. The concept of the nucleolus has been known since $1969$~\cite{schmeidler1969nucleolus}, and the question was posed as an important open problem in multiple papers. In $1998$, Faigle, Kern, Fekete, and Hochst{\"a}ttler~\cite{faigle1998nucleon} mention the problem in their work on the nucleon, a multiplicative-error analog to the nucleolus which they show is polynomial time computable. Kern and Paulusma state the question of computing the nucleolus for general matching games as an important open problem in $2003$~\cite{Kern2003}. In $2008$, Deng and Fang~\cite{deng2008algorithmic} conjectured this problem to be NP-hard, and in $2017$ Bir{\'{o}}, Kern, Paulusma, and Wojuteczky~\cite{Biro2017} reaffirmed this problem as an interesting open question. Theorem~\ref{th:main-result} settles the question, providing a polynomial-time algorithm to compute the nucleolus of a general instance of a weighted cooperative matching game.

Our approach to proving Theorem~\ref{th:main-result} is to provide a compact description of each feasible region polytope in a hierarchical sequence of linear programs described in the well-known {\em Maschler scheme}~\cite{Ko67,MPS79} for computing the nucleolus. While there are a linear number of LPs in the sequence, their naive implementation requires an exponential number of constraints. The base LP for this approach is the least core. It is known how to separate over the least core, but the challenge lies in solving all successive linear programs in the sequence.  Previous results in~\cite{KP03,Pa01} made use of the {\em unweighted} nature of node-weighted instances of matching games, and were able to employ the \emph{Edmonds-Gallai} structure theorem to derive compact formulations for the LPs in Maschler's hierarchy.  We do not know how to extend this line of work beyond node-weighted instances.

In our work, we identify a minimal family of tight excess constraints of (\ref{eq:coalition_leastcore}) corresponding to coalitions whose vertices are saturated by so called $\emph{universal matchings}$. A matching is universal if it saturates a coalition of vertices that is tight for all allocations in the least core. As we will show, universal matchings are the optimal matchings for carefully chosen cost functions derived from so called \emph{universal allocations} (see Section~\ref{subsec:uni-matching}). 
From here now, we rely on well-known characterizations of extreme points of matching polyhedra, and the fact that these 
are defined by laminar families of blossom constraints. Ultimately, the above allows us to obtain a decomposition of the input graph into the edges on either the inside or outside of blossoms. The structure of the associated optimum face of the matching polytope (e.g., see Schrijver~\cite{Sc02}) elucidates the structure of excess over all coalitions in the matching game. Our proof uses a critical insight into the symmetric nature of exchange on the nodes of a blossom as we move between leastcore allocations (see Lemma~\ref{lem:leastcore_properties}). We present the details of the leastcore LP in Section~\ref{sec:leastcore}. We show how this formulation can be used in Maschler's framework to compute the nucleolus in Section~\ref{sec:maschler}.

Prior to our work, the nucleolus was known to be polynomial-time computable only in structured instances of the matching game. Solymosi and Raghavan~\cite{SR94}
showed how to compute the nucleolus in an (unweighted) assignment game
instance in polynomial time. 
Kern and Paulusma~\cite{KP03} later
provided an efficient algorithm to compute the nucleolus in general
unweighted matching game instances. Paulusma~\cite{Pa01} extended the
work in~\cite{KP03} and gave an efficient algorithm to compute the
nucleolus in matching games where edge weights are induced by node
potentials. Farczadi~\cite{Fa15} finally extended Paulusma's framework further
using the concept of {\em extendible allocations}. 
We note also that it is easy to compute the nucleolus in weighted instances
of the matching game with non-empty core. For such instances, the
leastcore has a simple compact description that does not include
constraints for coalitions of size greater than $2$. Thus it is relatively straightforward to adapt the iterative algorithm of Maschler~\cite{MPS79} to a polynomial-time algorithm for computing the nucleolus (e.g., see~\cite[Chapter 2.3]{Fa15} for the details, Section~\ref{subsec:maschler} for an overview). 

\subsection{Related Work}

In a manner analogous to how we have defined matching games, a wide variety \emph{combinatorial optimization games} can be defined~\cite{DIN99}. In such games, the value of a coalition $S$ of players is succinctly given as the optimal solution to an underlying combinatorial optimization problem. It is natural to conjecture that the complexity of computing the nucleolus in an instance of such a game would fall in lock-step with the complexity of the underlying problem. Surprisingly this is not the case. For instance, computing the nucleolus is known to be NP-hard for network flow games~\cite{DFS09}, weighted threshold games~\cite{EG+07}, and spanning tree games~\cite{FKK98,FKP00}. On the other hand, polynomial time algorithms are known for finding the nucleolus of special cases of flow games, certain classes of matching games, fractional  matching games, and convex games~\cite{BH+10,chen2012computing,DFS09,FKK01,Fa15,GGZ98,GM+96,KP03,KSA00,Me78,Pa01,PRB06,SR94}.

One application of cooperative matching games is to network bargaining~\cite{EK10,Wi99}, where one considers a population of players that are interconnected by an underlying social network. Players engage in profitable {\em partnerships} -- each player with at most one of its neighbours. The profit generated through a partnership then needs to be shared between the participating players in an equitable way.  
Cook and Yamagishi~\cite{CY92} first proposed an elegant profit-sharing model that not only generalizes Nash's famous $2$-player bargaining solution~\cite{Na50}, but also validates empirical findings from the lab setting. 
Cook and Yamagishi's model defines a so called {\em outside option} for each player $v$ that essentially denotes the largest share that $v$ could demand in a partnership with any of its neighbours. An outcome is called {\em stable} if each player receives a profit share that is at least her outside option, and an outcome is {\em balanced} if value generated in excess of the sum of the outside options of the involved players is shared evenly. Kleinberg and Tardos~\cite{Kleinberg2008} showed that an instance of network bargaining has a balanced outcome if and only if it has a stable outcome. Bateni et al.~\cite{BH+10} later noted that stable outcomes for network bargaining correspond to elements of the core of the underlying cooperative matching game, and that balanced outcomes correspond to elements in the intersection of core and prekernel. 

It is well-known that the prekernel of a cooperative game may be
non-convex and that it may even be disconnected~\cite{Ko67,St68}. Despite this,
Faigle, Kern and Kuipers~\cite{FKK01} showed how to compute a point in
the intersection of prekernel and leastcore in polynomial time under
the reasonable assumption that the game has a polynomial time oracle
to compute the minimum excess coalition for a given allocation. Later
the same authors~\cite{faigle2006computing} refine their result to
computing a point in the intersection of the core and lexicographic
kernel, a set which is also known to contain the nucleolus. Bateni et
al. pose as an open question the existence of an efficiently
computable, balanced and {\em unique} way of sharing the profit. It is
well-known that the nucleolus lies in the intersection of core and
prekernel (if these are non-empty) and is always
unique~\cite{schmeidler1969nucleolus}. Theorem~\ref{th:main-result}
therefore resolves the latter open question left in~\cite{BH+10}.


Connections have also been made between matching games and stable matchings. Allocations in the core of the corresponding matching game have been found to be in 1-to-1 correspondence with stable matchings with payments, a variant of the Gale and Shapley's~\cite{gale1962college} famous stable marriage problem. See Koopmans and Beckmann~\cite{koopmans1957assignment}, Shapley and Shubik~\cite{SS71}, Eriksson and Karlander~\cite{eriksson2001stable}, and  Biro, Kern, and Paulusma~\cite{Biro2012} for details.

\subsection{Leastcore and Core of Matching Games.}\label{subsec:leastcore}

 It is straightforward to see that \eqref{eq:coalition_leastcore} can be rewritten equivalently as
\begin{align*}
\max ~ & \varepsilon \\
\text{s.t.} ~ & x(M) \geq w(M) + \varepsilon \qquad \text{for all}\quad M\in \cM \tag{{$P_1$}}\label{eq:matching_leastcore}\\
& x(V) = \nu(G)\\
& x \geq \0\,,
\end{align*}
where $\cM$ is the set of all matchings $M$ on $G$, and $x(M)$ is a shorthand for $x(V(M))$.

The separation problem for the linear
program~\eqref{eq:matching_leastcore} can be reduced to finding a
maximum weight matching in the graph $G$ with edge weights
$w(uv)-x(uv)$, $uv\in E$ (where we use $x(uv)$ as a shorthand for
$x(u)+x(v)$). Since the maximum weight matching can be found in
polynomial time~\cite{edmonds1965maximum}, we know that the linear
program~\eqref{eq:matching_leastcore} can be solved in polynomial time
as well~\cite{grotschel2012geometric}.

We use $\varepsilon_1$ to denote the optimal value of \eqref{eq:matching_leastcore} and $P_1(\varepsilon_1)$ for the set of allocations $x$ such that $(x,\varepsilon_1)$ is feasible for the leastcore linear program~\eqref{eq:matching_leastcore}.  In general, for a value $\varepsilon$ and a linear program~$Q$ on variables in $\R^V\times \R$ we denote by $Q(\varepsilon)$ the set $\{ x\in \R^V: (x,\varepsilon) \text{ is feasible for } Q\}$.

Note that $\varepsilon_1\leq 0$. Indeed, $\varepsilon \leq 0$ in any feasible solution $(x, \varepsilon)$ to \eqref{eq:matching_leastcore} as otherwise $x(M)$ would need to exceed $w(M)$ for all matchings $M$. In particular this would also hold for a maximum weight matching on $G$, implying that $x(V) > \nu(G)$. 
If $\varepsilon_1=0$ then the core of the cooperative matching game is non-empty. One can see that
$\varepsilon_1=0$ if and only if the value of a maximum weight
matching on $G$ with weights $w$ equals the value of a maximum weight
fractional matching. This follows since $x\in P_1(\varepsilon_1)$ is a
fractional weighted node cover of value $\nu(G)$ when
$\varepsilon_1=0$.  When $\varepsilon_1<0$, we say that the
cooperative matching game has an empty core. 
The matching game instance given in Appendix~\ref{appendix:example} 
has an empty core.

In this paper, we assume that the cooperative matching game $(G,w)$ has an empty core, as computing the nucleolus is otherwise well-known to be doable in polynomial time~\cite{Pa01}.

\subsection{Maschler's scheme}\label{subsec:maschler}

As discussed, our approach to proving Theorem~\ref{th:main-result}
relies on Maschler's scheme. The scheme requires us to solve
a linear number of LPs: $\{(P_j)\}_{j \geq 1}$ that we now
define. $(P_1)$ is the leastcore LP that we have already seen in Section~\ref{subsec:leastcore}. 
LPs $(P_j)$ for $j \geq 2$ are defined inductively. Crucial in their
definition is the notion of {\em fixed} coalitions that we introduce first. For a
polyhedron $Q \subseteq \R^V$ we denote by $\fix(Q)$ the collection of
sets $S\subseteq V$ such that $x(S)$ is constant over the polyhedron
$Q$, i.e.
$$
	\fix(Q):= \{ S \subseteq V\,:\quad x(S) = x'(S)\quad \text{for all}\quad x,x'\in Q\}\,.
$$
With this we are now ready to state LP $(P_j)$ for $j \geq 2$:
\begin{align*}
&\max \quad \varepsilon \tag{$P_j$}\label{eq:formulation_Maschler}\\
&\begin{array}{rcll}
\text{s.t.}\quad
x(S)-\nu(S) &\geq& \varepsilon &\text{for all}\quad S\subset V,\, S\not\in \fix(P_{j-1}(\varepsilon_{j-1}))\\
x &\in & P_{j-1}(\varepsilon_{j-1})\,, &
\end{array}
\end{align*}
where $\varepsilon_j$ be the optimal value of the linear
program~\eqref{eq:formulation_Maschler}. Let $j^*$ be the minimum
number $j$ such that $P_j(\varepsilon_j)$ contains a single point.
This point is the nucleolus of the game~\cite{Davis1965}. It is
well-known~\cite{MPS79} that
$P_{j-1}(\varepsilon_{j-1}) \subset P_{j}(\epsilon_j)$ and
$\varepsilon_{j-1} < \varepsilon_j$ for all $j$. Since the dimension
of the polytope describing feasible solutions
of~\eqref{eq:formulation_Maschler} decreases in every round until the
dimension becomes zero, we have $j^*\leq |V|$
\cite{MPS79}\cite[Pages~20-24]{Pa01}. 

Therefore, in order to find the nucleolus of the
cooperative matching game efficiently it suffices to solve
each linear program~\eqref{eq:formulation_Maschler}, $j=1,\ldots,j^*$
in polynomial time. We accomplish this by providing polynomial-size
formulations for~\eqref{eq:formulation_Maschler} for all $j \geq 1$. 

In Section~\ref{sec:leastcore} we introduce the concept of
\emph{universal matchings} which are fundamental to our approach, and
give a compact formulation for the first linear program in Maschler's
Scheme, the leastcore. We also present our main technical lemma, Lemma
\ref{lem:leastcore_properties}, which provides a crucial symmetry
condition on the values allocations can take over the vertices of
blossoms in the graph decomposition we use to describe the compact
formulation. In Section~\ref{sec:maschler} we describe the successive
linear programs in Maschler's Scheme and provide a compact formulation
for each one in a matching game.

\section{Leastcore Formulation.}\label{sec:leastcore}
\paragraph{}
In this section we provide a polynomial-size description of~\eqref{eq:matching_leastcore}. It will be useful to define a notation for \emph{excess}: for any $x \in P_1(\varepsilon_1)$ and $M \in \cM$ let $\ex(x,M) := x(M) - w(M)$.
\subsection{Universal Matchings, Universal Allocations.}
For each $x \in P_1(\varepsilon_1)$ we say that a matching $M\in \cM$ is an $x$-tight matching whenever $\ex(x,M)=\varepsilon_1$. We denote by $\cM^{x}$ the set of $x$-tight matchings.

A {\it universal matching} $M \in \cM$ is a matching which is $x$-tight for all $x \in P_1(\varepsilon_1)$. We denote the set of universal matchings on $G$ by $\Muni$. 
 A {\it universal allocation} $x^* \in P_1(\varepsilon_1)$ is a leastcore point whose $x^*$-tight matchings are precisely the set of universal matchings, i.e. $\cM^{x^*}=\Muni$.
\begin{lemma}\label{lemma:universal-allocations-exist}
There exists a universal allocation $x^* \in P_1(\varepsilon_1)$.
\end{lemma}
\begin{proof}
Indeed, it is straightforward to show that every $x^*$ in the relative interior (see \cite[Lemma 2.9(ii)]{ziegler2012lectures}) of $P_1(\varepsilon_1)$ is a universal allocation. If the relative interior is empty then $P_1(\varepsilon_1)$ is a singleton, which trivially contains a universal allocation. In Appendix~\ref{appendix-proofs} we provide a combinatorial proof.$\blacksquare$\end{proof} 
\begin{lemma}\label{lem:finding_universal_allocation}
A universal allocation $x^* \in P_1(\varepsilon_1)$ can be computed in  polynomial time.
\end{lemma}
\begin{proof}
A point $x^*$ in the relative interior of $P_1(\varepsilon_1)$ can be found in polynomial time using the ellipsoid method (Theorem 6.5.5~\cite{grotschel1988geometric},\cite{dadush2017}). Since any allocation $x^*$ from the relative interior of  $P_1(\varepsilon_1)$ is a universal allocation, this implies the statement of the lemma. 
$\blacksquare$\end{proof}
\paragraph{}
Given a non-universal allocation $x$ and a universal allocation $x^*$, we observe that $\cM^{x^*} \subset \cM^x$ and so $\theta(x^*)$ is strictly lexicographically greater than $\theta(x)$. Thus the nucleolus is a universal allocation. We emphasize that 
$\cM^{x^*}=\Muni$ is invariant under the (not necessarily unique) choice of universal
allocations $x^*$. Henceforth we fix a universal allocation~$x^* \in P_1(\varepsilon_1)$.

\subsection{Description for Convex Hull of Universal Matchings.}\label{subsec:uni-matching}
\paragraph{}
By the definition of universal allocation $x^*$, a matching $M$ is
universal if and only if it is $x^*$-tight. Thus, $M$ is a universal
matching if and only if its characteristic vector lies in the optimal
face of the matching polytope corresponding to (the maximization of)
the linear objective function assigning weight $-\ex(x^*,uv) = w(uv) - x^*(uv)$ to each
edge $uv \in E$. Let $\cO$ be the set of node sets $S \subseteq V$ such
that $|S| \geq 3$, $|S|$ is odd. Edmonds~\cite{edmonds1965maximum} gave a linear description of the matching
polytope of $G$ as the set of $y \in \R^{E}$ satisfying:
\begin{align*}
&\begin{array}{rcll}
y(\delta(v)) & \leq&1  \qquad&\text{for all} \quad v \in V\\
y(E(S)) &\leq& (|S| - 1)/2  &\text{for all}\quad S\in \cO \\
y &\geq & 0\,. &
\end{array}
\end{align*}
Thus, a matching $M\in \cM$ is universal if and only if it satisfies the constraints
\begin{equation}\label{eq:face_Muni}
\begin{array}{rcll}
M \cap \delta(v) &=&1  \qquad&\text{for all}\quad v \in W\\
M \cap E(S) &=& (|S| - 1)/2  &\text{for all}\quad S\in \cL \\
M\cap \{e\} &=& 0 &\text{for all}\quad e \in F\,,
\end{array}
\end{equation}
where $W$ is some subset of $V$, $\cL$ is a subset of $\cO$, and $F$
is a subset of $E$. Using an uncrossing argument, as in
\cite[Pages~141-150]{lau2011iterative}, we may assume that the
collection of sets $\cL$ is a laminar family of node sets; i.e., for
any two distinct sets $S,T\in\cL$, either $S\cap T =\varnothing$ or
$S\subseteq T$ or $T\subseteq S$.

\begin{lemma}\label{lemma:inessential-tight}
For every node $v\in V$ there exists $M\in\Muni$ such that $v$ is exposed by $M$. Hence, $W=\varnothing$.
\end{lemma}
\begin{proof}
Assume for a contradiction that there exists a node $v\in V$ such that $v\in W$.

First, note that there always exists a non-universal matching $M \in
\cM \setminus \Muni$ since otherwise the empty matching would be
universal, and thus 
\[ 0 = x^*(\varnothing) = w(\varnothing) + \varepsilon_1, \]
implying that the core of the given matching game instance is non-empty. 

Suppose first that there
exists a node $u\in V$ exposed by some matching $M' \in \Muni$ such that $x^*_u>0$. We define
$$
	\delta_0:=\min\{\ex(x^*,M)-\varepsilon_1\,:\quad M \in \cM\setminus \Muni\}\,.
$$
Recall that $\Muni$ is the set of maximum weight matchings on $G$ with
respect to the node weights $w(uv)-x^*(uv)$, $uv\in E$, i.e. $\Muni$ is
the set of $x^*$-tight matchings. Moreover, recall that
$\ex(x^*,M)=\varepsilon_1$ for $M\in \Muni$. Thus, we have
$\delta_0>0$.

We define $\delta:=\min\{\delta_0, x^*_u\}>0$ and a new allocation $x'$ as follows:
$$
x'_r := \begin{cases}
x^*_r + \delta, &\text{if $r=v$} \\
x^*_r - \delta, &\text{if $r=u$} \\
x^*_r, &\text{otherwise}.
\end{cases}$$
Since all universal matchings contain $v$, the excess with respect to $x'$ of any universal
matching is no smaller than their excess with respect to $x^*$. Therefore, by our choice
of $\delta$, $(x',\epsilon_1)$ is a feasible, and hence optimal, solution for
\eqref{eq:matching_leastcore}. But $M'$ is not $x'$-tight, since $M'$ covers $v$ and exposes $u$. This contradicts that $M'$ is a universal matching.

Now consider the other case: for all $u\in V$ if $u$ is exposed by a universal matching then $x^*_u=0$. Then, for every universal matching $M\in \Muni$ we have
$$
\varepsilon_1=\ex(x^*,M)=x^*(V)-w(M)=\nu(G)-w(M).
$$
Since $\nu(G)$ is the maximum weight of a matching on $G$ with respect to the weights $w$, we get that $\varepsilon_1\geq 0$. Thus $x^*$ is in the core, contradicting our assumption that the core is empty.
$\blacksquare$\end{proof}

\subsection{Description of Leastcore.}


We  denote inclusion-wise maximal sets in the family~$\cL$ as
$
	S_1^*,S_2^*,\ldots, S_k^*.
$
We define the edge set $E^+$ to be the set of edges in $G$ such that at most one of its nodes is in $S_i^*$ for every $i\in [k] :=\{1, \dots, k\}$, i.e.
$$
	E^+:=E\setminus \big(\bigcup_{i=1}^{k} E(S_i^*)\big).
$$
\begin{lemma}\label{lemma:expose}
For every choice of $v_i \in S_i^*$, $i \in [k]$, there exists a universal matching $M \in \Muni$ such that the node set covered by $M$ is as follows
$$\bigcup_{i=1}^k S_i^* \backslash\{v_i\}\,.$$
\end{lemma}
\begin{proof}
By Lemma~\ref{lemma:inessential-tight}, we know that for every $i \in [k]$ there exists a universal matching $M_{v_i}\in \Muni$ such that $v_i$ is exposed by $M_{v_i}$. Now, for every $i \in [k]$, let us define 
$$
	M_i:=E(S_i^*)\cap M_{v_i}.
$$ 
Since $M_i$ satisfies all laminar family constraints in $\cL$ for subsets of $S_i^*$ we have that
$$
	\bigcup_{i=1}^k M_i
$$
is a matching satisfying all the constraints~\eqref{eq:face_Muni}, and hence is a universal matching covering the desired nodes.
$\blacksquare$\end{proof}

For each $i \in [k]$ fix a unique \emph{representative} node $v_i^* \in S_i^*$. By Lemma~\ref{lemma:expose}, there exists a universal matching $M^*$ covering precisely $\bigcup_{i \in [k]}S_i^*\backslash\{v_i^*\}$. For any $x \in P_1(\varepsilon_1)$ and $S\subseteq V$ we use $\sym(x, S)$ to denote $$\sym(x,S):=x(S) - x^*(S)\,.$$ For single nodes we use the shorthand $\sym(x,v) = \sym(x,\{v\})$. We now prove the following crucial structure result on allocations in the leastcore.

\begin{lemma}\label{lem:leastcore_properties}
For every leastcore allocation $x$, i.e. for every $x\in P_1(\varepsilon_1)$, we have that
	\begin{enumerate}[(i)]
		\item \label{property:equality} for all $i \in [k]$, for all $u\in S_i^*$:$\quad \sym(x,u)=\sym(x,v_i^*),$

		\item \label{property:nonnegativity}for all $e\in E^+$:$\quad \ex(x,e)\geq 0.$
	\end{enumerate}
\end{lemma}

\begin{proof}
  Consider $u\in S^*_i$, and note that
  we may use Lemma~\ref{lemma:expose} to choose a universal matching
  $M_u$ covering precisely $$S_i^*\backslash\{u\} \cup \bigcup_{j \neq i} S_j^*\backslash\{v_j^*\}.$$ 
 Hence we have $V(M_u) \cup \{u\} = V(M^*) \cup\{v_i^*\}$, and since $M^*$ and $M_u$ are universal, $x(M^*) = x^*(M^*)$ and $x(M_u) = x^*(M_u)$. Using these observations we see that
 $$\sym(x,u) = x(u) + x(M_u) - (x^*(u) - x^*(M_u)) = x(v_i^*) +x(M^*) -(x^*(v_i^*) - x^*(M^*)) = \sym(x,v_i^*).$$
  showing \eqref{property:equality}.

\smallskip

Now we prove~\eqref{property:nonnegativity}. Consider $e\in E^+$ where $e=\{u,v\}$. Since $e \not\in E(S_i^*)$ for all $i \in [k]$, we can choose a universal matching $M$ exposing $u$ and $v$ by Lemma~\ref{lemma:expose}. Thus $M\cup\{e\}$ is also a matching. Notice $M$ is $x$-tight, and so we have
$$
	\ex(x,e)=\underbrace{\ex(x,M\cup\{e\})}_{\geq \varepsilon_1}-\underbrace{\ex(x,M)}_{= \varepsilon_1}\geq 0\
$$
as desired.
$\blacksquare$\end{proof}

\begin{lemma}\label{lem:cardinality-matching}
Let $x \in P_1(\varepsilon_1)$ and let $M \in \cM$ be a matching such that $M \subseteq \bigcup_{i\in[k]}E(S_i^*)$. Then there exists $M' \subseteq M^*$ such that $\ex(x,M') \leq \ex(x,M)$ and for all $i \in [k]$, $|M' \cap E(S_i^*)| = |M\cap E(S_i^*)|$.
\end{lemma}
\begin{proof}
 See Appendix~\ref{appendix-proofs}.
$\blacksquare$\end{proof}
\bigskip
Recall that $x^*$ is a fixed universal allocation in
$P_1(\varepsilon_1)$. Let $E^*\subseteq E$ denote the union of universal matchings, i.e. $E^* = \cup_{M \in \Muni} M$. We now define linear
program~\eqref{eq:leastcore_compact}. 
\begin{align}
\max \quad & \varepsilon \tag{$\overline{P}_1$}\label{eq:leastcore_compact}\\
\text{s.t.} \quad & \sym(x,u) =\sym(x,v_i^*)  && \text{for all} \quad
                                                  u \in S_i^*,\, i
                                                  \in [k] \label{constr:symmetry} \\
& \ex(x,e) \leq 0 && \text{for all }\quad e \in E^* \notag\\
& \ex(x,e) \geq 0 && \text{for all}\quad e\in E^+\notag \\
& \ex(x,M^*) = \varepsilon \notag \\
& x(V)=\nu(G) \notag \\
& x \geq 0\,. \notag
\end{align}
Let $\overline{\varepsilon}_1$ be the optimal value of the linear
program~\eqref{eq:leastcore_compact}. 
We now show that
$\overline{P_1}(\overline{\varepsilon}_1)$ is indeed a compact description of the
leastcore $P_1(\varepsilon_1)$.

\begin{theorem}\label{thm:compact_leastcore}
We have $\varepsilon_1=\overline{\varepsilon}_1$ and $P_1(\varepsilon_1)=\overline{P}_1(\overline{\varepsilon}_1)$.
\end{theorem}
\begin{proof}
First, we show that $P_1(\varepsilon_1)\subseteq
\overline{P}_1(\varepsilon_1)$. Consider $x\in P_1(\varepsilon_1)$. By Lemma~\ref{lem:leastcore_properties}\eqref{property:equality} we have
$$
	\sym(x,u) =\sym(x,v_i^*)  \qquad\text{for all} \quad u \in S_i^*,\, i \in [k].
$$
Lemma~\ref{lem:leastcore_properties}\eqref{property:nonnegativity}
shows that
$\ex(x,e) \geq 0$ for all $e\in E^+$, and
$\ex(x,M^*)=\varepsilon_1\,$ holds by
the universality of $M^*$.
It remains to show that 
$$
\ex(x,e) \leq 0 \quad\text{for all }\quad e \in E^*.
$$
Suppose for contradiction there exists $e \in E^*$ such that $\ex(x,e) > 0$. By the definition of $E^*$, there exists a universal matching $M'$ containing $e$. Since $M'$ is universal, $\ex(x,M') = \varepsilon_1$. But by our choice of $e$,
$$\ex(x,M'\setminus\{e\})<\ex(x,M')= \varepsilon_1$$
contradicting that $x$ is in $P_1(\varepsilon_1)$. Thus we showed that $(x,\varepsilon_1)$ is feasible for~\eqref{eq:leastcore_compact}, i.e. we showed that $P_1(\varepsilon_1)\subseteq \overline{P}_1(\varepsilon_1)$.

To complete the proof we show that $\overline{P}_1(\overline{\varepsilon}_1)\subseteq P_1(\overline{\varepsilon}_1)$. Let $x$ be an allocation in $\overline{P}_1(\overline{\varepsilon}_1)$. Due to the description of the linear program~\eqref{eq:matching_leastcore}, it is enough to show that for every matching $M\in \cM$ we have
$$
	\ex(x,M)\geq \overline{\varepsilon}_1\,.
$$
Since 
$
\ex(x,e) \geq 0 \text{ for all} e\in E^+,
$
it suffices to consider only the matchings $M$, which are unions of matchings on the graphs $G[S_i^*]$, $i \in [k]$. Let $t_i:=|M\cap E(S_i^*)|$. By Lemma~\ref{lem:cardinality-matching} applied to $x^*$ there exists $M' \subseteq M^*$ such that $\ex(x^*,M) \geq \ex(x^*, M')$ and $|M'\cap E(S_i^*)|=t_i$, for all $i\in [k]$. Then due to constraints (\ref{constr:symmetry}) in~\eqref{eq:leastcore_compact}
we have
$$\ex(x,M) = \underbrace{\sum_{i=1}^k 2t_i\sym(x,v_i^*)}_{=\sym(x,M')} + \underbrace{\ex(x^*, M)}_{\geq \ex(x^*,M')} \geq \ex(x,M') \geq \ex(x,M^*) = \overline{\varepsilon}_1\,,$$
where the last inequality follows since $M' \subseteq M^*$ and $\ex(x,e) \leq 0$ for all $e \in E^*$.
\bigskip

Thus, we showed that $P_1(\varepsilon_1)\subseteq \overline{P}_1(\varepsilon_1)$ and $\overline{P}_1(\overline{\varepsilon}_1)\subseteq P_1(\overline{\varepsilon}_1)$. Recall, that $\varepsilon_1$ and $\overline{\varepsilon}_1$ are the optimal values of the linear programs~\eqref{eq:matching_leastcore} and~\eqref{eq:leastcore_compact} respectively. Thus, we have $\varepsilon_1=\overline{\varepsilon}_1$ and 
$P_1(\varepsilon_1)=\overline{P}_1(\overline{\varepsilon}_1)$.
$\blacksquare$\end{proof}

\section{Computing the Nucleolus} \label{sec:maschler}

The last section presented a polynomial-size formulation for the
leastcore LP (\ref{eq:matching_leastcore}). 
In this section we complete our polynomial-time implementation of
Maschler's scheme by showing that
\eqref{eq:formulation_Maschler} has the following compact
reformulation:

\begin{align*}
&\max \quad \varepsilon \tag{$\overline{P}_j$}\label{eq:simple_formulation_Maschler}\\
&\begin{array}{rcll}
\text{s.t.}\quad
\ex(x,e) &\geq& \varepsilon-\varepsilon_1 &\text{for all}\quad e\in E^+,\, e\not\in  \fix(\overline{P}_{j-1}(\overline{\varepsilon}_{j-1}))\\
x(v)&\geq&\varepsilon-\varepsilon_1 &\text{for all}\quad v\in V,\, v\not\in \fix(\overline{P}_{j-1}(\overline{\varepsilon}_{j-1}))\\
\ex(x,e)&\leq & \varepsilon_1-\varepsilon  \qquad&\text{for all} \quad e \in E^*,\, e\not\in  \fix(\overline{P}_{j-1}(\overline{\varepsilon}_{j-1}))\\
x &\in & \overline{P}_{j-1}(\overline{\varepsilon}_{j-1})\,, &
\end{array}
\end{align*}
where $\overline{\varepsilon}_j$ is the optimal value of the linear program~\eqref{eq:simple_formulation_Maschler}
\begin{theorem}\label{thm:nucleolus_simple}
For all $j=1,\ldots,j^*$, we have $\varepsilon_j=\overline{\varepsilon}_j$ and $P_j(\varepsilon_j) = \overline{P}_j(\overline{\varepsilon}_j)$.
\end{theorem}

\begin{proof}
We proceed by induction. For $j=1$ the statement holds due to Theorem~\ref{thm:compact_leastcore}. Let us show that statement holds for each $j=2,\dots, j^*$, assuming that the statement holds for $j-1$. By induction, $P_{j-1}(\varepsilon_{j-1}) = \overline{P}_{j-1}(\overline{\varepsilon}_{j-1})$. We let $\cF := \fix(P_{j-1}(\varepsilon_{j-1})) = \fix(\overline{P}_{j-1}(\overline{\varepsilon}_{j-1}))$ to ease presentation. 

First we show $P_j(\varepsilon_j) \subseteq \overline{P}_j(\varepsilon_j)$. Let $x \in P_j(\varepsilon_j)$. First consider an edge $e \in E^+\backslash\cF$. By Lemma~\ref{lemma:expose} there exists a universal matching $M \in \Muni$ exposing the endpoints of $e$. Moreover, since $V(M) \in \fix(P_1(\varepsilon_1)) \subseteq \cF$ we have $V(M\cup \{e\}) \not\in \cF$. Thus we see that
$$\ex(x,e) = \underbrace{\ex(x,M\cup\{e\})}_{\geq
    \varepsilon_j}-\underbrace{\ex(x,M)}_{=\varepsilon_1}\geq\varepsilon_j-\varepsilon_1.$$
Now consider $v \in V \backslash \cF$. Similar to the previous argument, but via Lemma~\ref{lemma:inessential-tight}, we can find a universal matching $M$ exposing $v$. Thus, since $M$ is universal, $V(M) \cup \{v\} \not\in \cF$. We also have $\nu(V(M) \cup\{v\}) \geq w(M)$, so
$$x(v) = x(M\cup\{v\}) - w(M)- \ex(x,M) \geq \varepsilon_j - \varepsilon_1.$$ Finally consider $e=\{u,v\} \in E^*\backslash \cF$. Since $e \in E^*$ there exists a universal matching $M \in \Muni$ such that $e$ is in $M$. Since $V(M) \in \fix(P_1(\varepsilon_1)) \subseteq \cF$ and $\{u,v\} \not\in \cF$, we see that $M \setminus \{e\} \not\in \cF$. Thus we have
$$\ex(x,e) = \ex(x,M) - \ex(x,M\backslash\{e\}) \leq \varepsilon_1 -  \varepsilon_j$$
as desired.

So we have shown $P_j(\varepsilon_j) \subseteq \overline{P}_j(\varepsilon_j)$, and thus it remains to show $\overline{P}_j(\overline{\varepsilon}_j) \subseteq P_j(\overline{\varepsilon}_j)$. Let $x \in \overline{P}_j(\overline{\varepsilon}_j)$ and let $S \subset V$ such that $S \not\in \cF$. We need to show
$$x(S) - \nu(S) \geq \overline{\varepsilon}_j.$$
Since $S \not\in \cF$, there exists $v \in S$ such that $\{v\} \not\in \cF$. Let $M$ be a maximum $w$-weight matching on $G[S]$, i.e. $w(M) = \nu(S)$. Either
$v \in S \backslash V(M)$ or $v \in V(M)$. We proceed by case
distinction.

\underline{{\bf Case 1:} $v \in S\backslash V(M)$.} Since $x \geq 0$ we have
$x(S) \geq x(V(M) \cup \{v\})$, and it suffices to prove that the
right-hand side exceeds the weight of $M$ by at least
$\overline{\varepsilon}_j$. 
By assumption,  $x \in \overline{P}_j(\overline{\varepsilon}_j) \subseteq
\overline{P}_1(\overline{\varepsilon}_1) = P_1(\varepsilon_1)$, and hence
\begin{equation}\label{eq:M}
  \ex(x,M) \geq \epsilon_1.
\end{equation}
Together with  $x(v)\geq \overline{\varepsilon}_j-\varepsilon_1$ this yields the desired inequality. 

\underline{{\bf Case 2:} $v \in V(M)$.}  The same argument as before applies if
there is $u \in S\backslash V(M)$ with $\{u\} \not\in
\cF$. Therefore, we focus on the case where
$\{u\} \in \cF$ for all $u \in
S\backslash V(M)$. Then $M$ contains an
edge $f$ such that $f \not\in \cF$. Again
using $x\geq 0$, it suffices to show that $x(V(M)) \geq \overline{\epsilon}_j$.
We further distinguish
cases: either $f \in M \cap E^+$ or $f \in M\backslash E^+$.

{\bf Case 2a}: $f \in M \cap E^+$. Here, the desired inequality follows from $\ex(x,M\backslash\{f\})\geq \varepsilon_1$ due to $x \in P_1(\varepsilon_1)$ and from
$\ex(x,f)\geq \overline{\varepsilon}_j-\varepsilon_1$.

{\bf Case 2b}: $f \in M \backslash E^+$. If $M \cap E^+$ has an edge
that is not in $\cF$ then we use the same
argument as in Case 2a. So we may assume that all of the edges in
$M \cap E^+$ are in $\cF$. Now recall that $x \in P_1(\varepsilon_1)$, and
hence $\ex(x,e) \geq 0$ for all $e \in E^+$ by Lemma
\ref{lem:leastcore_properties}.\eqref{property:nonnegativity}. Thus, we may assume that $V(M)\not\in \cF$ and $M\cap E^+=\varnothing$. 

For each $i \in [k]$ choose some node $u_i \in S_i^*$ exposed by $M$. By Lemma~\ref{lemma:expose} there exists $\hat{M} \in \Muni$ such that $\hat{M}$ exposes precisely $u_1, \dots, u_k$ on $\cL$. Thus $V(M) \subseteq V(\hat{M})$ and we see 
$$x(M) - x(\hat{M}) = x(V(\hat{M})\backslash V(M)) = \sum_{i=1}^k x((V(\hat{M})\backslash V(M)) \cap S_i^*).$$
Since $V(M) \not\in \cF$, and $V(\hat{M}) \in \cF$ this equation implies that there exists $i\in [k]$ such that $|M \cap E(S_i^*)| < \frac{|S_i^*| - 1}{2}$ and $(V(\hat{M})\backslash V( M)) \cap S_i^* \not\in\cF$. So there exists $v \in V(\hat{M}) \cap S_i^*$ such that $\{v\} \not\in \cF$. Since $\sym(x,v) = \sym(x,v_i^*)$, we have that $\{v_i^*\}\not\in\cF$. By Lemma~\ref{lem:cardinality-matching}, applied to $x$, there exists $M' \subseteq M^*$ such that $|M' \cap E(S_i^*)| = |M\cap E(S_i^*)| < |M^* \cap E(S_i^*)|$ and $\ex(x,M) \geq \ex(x,M')$. Note, that since $|M' \cap E(S_i^*)| = |M\cap E(S_i^*)|$, $V(M)\subseteq \cup_{i=1}^k S^*_i$ and $V(M)\not\in \cF$, we have $V(M')\not\in \cF$. Choose $e \in M^* \backslash M'$, $e\not\in \cF$, then we have $e \in M^* \backslash M' \subseteq E^*$. Hence,
$$\ex(x,M)\geq \ex(x,M') = \ex(x,M^*)-\ex(x,e) -\ex(x,M^* \backslash (M'\cup\{e\}))\geq \varepsilon_1 -(\varepsilon_1 - \overline{\varepsilon}_j) = \overline{\varepsilon}_j,$$
where the last inequality follows since $\ex(x,M^* \backslash (M'\cup\{e\}))\leq 0$ and  $x \in \overline{P}_{j}(\overline{\varepsilon}_j)$, $e \in E^*\backslash\cF$. 
$\blacksquare$\end{proof} \bigskip With Theorem
\ref{thm:nucleolus_simple} we can replace each linear program~\eqref{eq:formulation_Maschler}
with~\eqref{eq:simple_formulation_Maschler} in Maschler's Scheme. Since the universal
allocation $x^*$, the node sets $S^*_i$, $i\in[k]$, the edge set $E^+$, and the edge set $E^*$ can all be computed in polynomial time, we have shown that the nucleolus of any cooperative matching game with empty core can be computed in polynomial time. Therefore we have shown Theorem~\ref{th:main-result}. 

\paragraph{Open Questions} 

Matching Games generalize naturally to $b$-matching games, where instead the underlying optimization problem is to find an edge subset $M$ with $|M\cap \delta(v)| \leq b_v$ for each node $v$. Biro, Kern, Paulusma, and Wojuteczky~\cite{Biro2017} showed that the core-separation problem when $b_v > 2$ for some vertex $v$, is coNP-Hard. Despite this, the complexity of computing the nucleolus of these games is open.


Our algorithm relies heavily on the ellipsoid method. When the core is non-empty, there is a combinatorial algorithm for finding the nucleolus~\cite{Biro2012}. It would be interesting to develop a combinatorial algorithm for nucleolus computation of matching games in general.

\paragraph*{Acknowledgements.}

The authors thank Umang Bhaskar, Daniel Dadush, and Linda Farczadi for stimulating and insightful discussions related to this paper.
\newpage

\bibliography{references}
\bibliographystyle{splncs04}

\newpage

\appendix
\section*{Appendix}

\section{Omitted Proofs}\label{appendix-proofs}
\begin{proof}
(Of Lemma~\ref{lemma:universal-allocations-exist})
Let $\cM^x$ denote the set of $x$-tight matchings for some $x \in P_1(\varepsilon_1)$. That is $\cM^x:=\{ M \in \cM : \ex(x,M) = \varepsilon_1\}$. We claim that there exists $x^* \in P_1(\varepsilon_1)$ such that 
$$\cM^{x^*} = \bigcap_{x \in P_1(\varepsilon_1)} \cM^x\,.$$
 Let $x^* \in P_1(\varepsilon_1)$ be chosen to minimize $|\cM^{x^*}|$. Suppose for a contradiction that $x^*$ is not universal. Then there exists $x \in P_1(\varepsilon_1)$ and $M \in \cM^{x^*}\backslash \cM^{x}$. Let $\bar{x} := \frac{1}{2}(x^* + x)$. Since $P_1(\varepsilon_1)$ is a convex set, $\bar{x} \in P_1(\varepsilon_1)$. Furthermore, $\cM^{\bar{x}} = \cM^{x^*} \cap \cM^{x}$, thus $\cM^{\bar{x}} \subseteq \cM^{x^*} \backslash \{M\}$ contradicting the minimality of $|\cM^{x^*}|$.
$\blacksquare$
\end{proof}

\paragraph{Optimal Matchings of Restricted Cardinality}
 It is well-known\cite{Sc02}, that for any $t \in \N$ and for any graph $H$, $\conv\{M \in \cM(H) : |M| = t\}$ has the following linear description:
\begin{align*}
   P_t(H) := \{ x \in \R^E &: \\
    \ &x(\delta(v)) \leq 1 &\text{for all}\quad  v \in V(H) \\
    \ &x(E(U)) \leq \frac{|U| - 1}{2} &\text{for all}\quad U \in \cO(H) \\
    \ &x(E(H)) =t \\
    \ &x\geq 0\}.
\end{align*}
For any given $c \in R^{E(H)}$ we denote by $P_t^c(H)$ the set of vertices $x$ of the above polytope maximizing $c^Tx$, i.e. the optimal solutions to the linear program $\max\{c^Tx : x \in P_t(H)\}$. The dual to this linear optimization problem is to minimize
$$\sum_{v\in V(H)} y_v + \sum_{U \in \cO(H)} \frac{|U|-1}{2}z_U + t\gamma $$
where $(y,z,\gamma)$ is in $D_{t,c}(H)$ defined as follows:
\begin{align*}
    D_{t,c}(H) := \{(y,z,\gamma) \in \R^{V(H)} \times \R^{\cO(H)} \times \R &:\\
    \ &y_u + y_v + \gamma + \sum_{U \in \cO : uv \subseteq U} z_U \geq c(uv) &\text{for all}\quad uv \in E(H)\\
    \ &y,z \geq 0\}.
\end{align*}
Let $D_t^c(H)$ denote the set of optimal solutions to $\min\{\sum_{v\in V(H)} y_v + \sum_{U \in \cO(H)} \frac{|U|-1}{2}z_U + t\gamma: (y,z,\gamma)\in D_{t,c}(H)\}.$ 
\begin{lemma}\label{lemma:cardinality-matching}
Suppose $2 \leq t \leq \floor{\frac{|V(H)|}{2}}$. Let $x \in P_t^c(H)$ and $(y,z,\gamma) \in D_t^c(H)$. If the support of $z$ is laminar and $y=0$ then there exists $e \in \supp(x)$ such that
$$x - \chi(e) \in P_{t-1}^c(H)$$
and
there exists $(0,z',\gamma') \in D_{t-1}^c(H)$ with the support of $z'$ laminar.
\end{lemma}

\begin{proof} Let $\cL = \supp(z)$ be the laminar family defined by the support of $z$. Let $S_1, \dots, S_\ell \in \cL$ be the top level sets of $\cL$ (i.e. containment maximal sets), ordered so that
$$0 < z_{S_1} \leq z_{S_2} \leq \dots \leq z_{S_\ell}.$$
Since $y = 0$, if there exists $e \in \supp(x) \cap (E(H)\backslash\bigcup_{i \in [\ell]}E(S_i))$ then by complementary slackness, $x-\chi(e) \in P_{t-1}^c(H)$ and $(0,z,\gamma) \in D_{t-1}^c(H)$. Thus we may assume that $x(uv) = 0$ for all $uv \in E(H)\backslash \bigcup_{i \in [\ell]} E(S_i)$. Let $uv \in E(S_1)$ such that $x_{uv} = 1$ and $S_1$ is a minimal set in $\cL$ containing $uv$. Complementary slackness assures us that such $uv$ exists. Let $x' = x - \chi(uv)$. Define $z' \in \R^{\cO(H)}$ as follows
$$ z'_U = \begin{cases}
z_{U} - z_{S_1}, &\text{if $U = S_i$ for some $i \in [\ell]$} \\
z_{U}, &\text{otherwise}.
\end{cases}$$
Define $\gamma' = \gamma + z_{S_1}$. We will show $x' \in P_{t-1}^c(H)$ and $(0,z',\gamma') \in D_{t-1}^c(H)$. First we verify feasibility. Indeed, since $x$ is a matching with $t$ edges, $x'$ is a matching with $t-1$ edges and primal feasibility is satisfied. For dual feasibility, observe by our choice of $S_1$ that $z' \geq 0$, and for any edge $uv \in \bigcup_{i \in [\ell]} E(S_i)$ the net effect on the left hand side of the dual constraint associated with $e$ is $0$. For $uv \in E(H)\backslash \bigcup_{i \in [\ell]}E(S_i)$, the left hand side of the dual constraint associated with $e$ increased by $z_{S_1}$.
\paragraph{}
Clearly since $\supp(z)$ is laminar, and $\supp(z') = \supp(z) \backslash \{S_1\}$, we have that $\supp(z')$ is laminar. It remains to verify optimality, which we will show via complementary slackness. Consider some $uv \in E(H)$ for which $x'_{uv} > 0$. Then $uv \in E(S_i)$ for exactly one $i$. Hence
$$\gamma' + \sum_{U \in \cO(G) : uv \subseteq U} z'_U = \gamma + z_{S_1} + (-z_{S_1}) + \sum_{U \in \cO(G) : uv \subseteq U} z_U  = c(uv)$$
where the last equality follows from complementary slackness for $x$ and $(0,z,\gamma)$. Lastly, let $U \in \cL$ such that $z'_U > 0$. By our choice of $z'$, $U \neq S_1$ and so
$$x'(E(U)) = x(E(U)) = \frac{|U|-1}{2}$$
where the last equality follows from complementary slackness for $x$ and $(0,z,\gamma)$.
$\blacksquare$\end{proof}

\begin{proof}
(Of $\textbf{Lemma}$~\ref{lem:cardinality-matching}) Let $i \in [k]$ and let $H = G[S_i^*]$. Let $t = |M^*\cap E(H)|$, and let $c \in \R^{E(H)}$ be defined by $c(uv) = w(uv) - x(uv)$ for all $uv \in E(H)$. Let $(y,z,\gamma) \in D_{t}^{c}(H)$. Since $c^T\chi(M) = -\ex(x,M)$ for all $M \in \cM(H)$, by Lemma~\ref{lemma:inessential-tight} and complementary slackness, we have $y=0$. Now we use standard uncrossing techniques (see for ex.~\cite[Pages~141-150]{lau2011iterative}) to obtain $(0,z,\gamma)\in D_t^c(H)$ with $\supp(z)$ laminar. Note that $M^*\cap E(H)$ is in $P_t^c(H)$, otherwise we can replace $M^*\cap E(H)$, within $H$, with an optimal matching in $P_t^c(H)$ to obtain a matching with lower excess than $M^*$, contradicting the universal tightness of $M^*$. Apply Lemma~\ref{lemma:cardinality-matching} inductively to obtain $M'_i \subseteq M^*\cap E(H)$ such that $\chi(M'_i) \in P_{|M'\cap E(H)|}^c(H)$. Since all $S_i^*$ are disjoint, the matching $M' = \bigcup_{i \in [k]} M'_i$ is the desired matching.
$\blacksquare$\end{proof}

\section{Example of a Matching Game With Empty Core}\label{appendix:example}

Consider the example in Fig.~\ref{fig:example}. This graph $G=(V,E)$ is a $5$-cycle with two adjacent edges $15$ and $12$ of weight $2$, and the remaining three edges of weight $1$. Since the maximum weight matching value is $\nu(G) = 3$, but the maximum weight fractional matching value is $\frac{7}{2}$, the core of this game is empty. The allocation $x^*$ defined by $x^*(1) = \frac{7}{5}$ and $x^*(2) = x^*(3) = x^*(4) = x^*(5) = \frac{2}{5}$ lies in the leastcore. Each edge has the same excess, $-\frac{1}{5}$, and any coalition of four vertices yields a minimum excess coalition with excess $-\frac{2}{5}$. Hence the leastcore value of this game is $\varepsilon_1=-\frac{2}{5}$.

\begin{wrapfigure}{r}{0.25\textwidth}
	\begin{center}
		\includegraphics[width=0.2\textwidth]{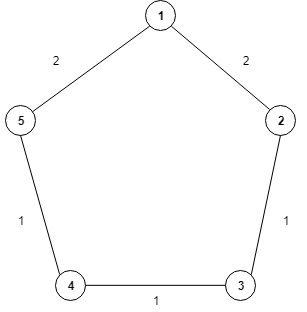}
	\end{center}
	\caption{Matching Game with Empty Core}\label{fig:example}
\end{wrapfigure}
In fact, we can see that $x^*$ is the nucleolus of this game. To certify this we can use the result of Schmeidler~\cite{schmeidler1969nucleolus} that the nucleolus lies in the intersection of the leastcore and the prekernel. For this example, the prekernel condition that for all $i\neq j \in V$,
$$ \max_{S\subseteq V\backslash\{j\}} x(S\cup\{i\}) -\nu(S\cup\{i\}) = \max_{S\subseteq V\backslash\{i\}} x(S\cup\{j\}) -\nu(S\cup\{j\})$$ reduces to the condition that the excess values of non-adjacent edges are equal. Since $G$ is an odd cycle, this implies that all edges has equal excess, i.e.
$$\ex(x,12)=\ex(x,23)=\ex(x,34)=\ex(x,45)=\ex(x,15).$$
Combining the four equations above with the leastcore condition that $x(V) = \nu(G)$
we obtain a system of equations with the unique solution $x^*$. Hence the intersection of the leastcore and prekernel is precisely $\{x^*\}$, and so by Schmeidler, $x^*$ is the nucleolus.

%
%
%

\end{document}